\documentclass[12pt,a4paper]{article}

\usepackage{authblk}
\usepackage{graphicx}
\usepackage{subcaption}
\usepackage[font=footnotesize]{caption}
\usepackage{cite}
\usepackage[cmex10]{amsmath}
\usepackage{amsfonts,amssymb,amsthm}
\usepackage{mathrsfs}
\usepackage{enumerate}
\usepackage{mathtools}
\usepackage{hyperref}

\newtheorem{proposition}{Proposition}
\newtheorem{corollary}{Corollary}
\newtheorem{theorem}{Theorem}
\newtheorem{lemma}{Lemma}

\theoremstyle{remark}
\newtheorem{remark}{Remark}

\newcommand{\qbinom}[2]{\genfrac{[}{]}{0pt}{}{\,#1\,}{\,#2\,}_{\!q}} 

\DeclareMathOperator{\Row}{\text{Row}}
\DeclareMathOperator{\Span}{\text{Span}}

\begin{document}

\title{Probability of Partially Decoding Network-Coded Messages}

\author[1]{Jessica~Claridge}
\author[2]{Ioannis~Chatzigeorgiou}
\affil[1]{Department of Mathematics, Royal Holloway, University of London, UK\authorcr(e-mail: jessica.claridge.2013@live.rhul.ac.uk)}
\affil[2]{School of Computing and Communications, Lancaster University, UK\authorcr(e-mail: i.chatzigeorgiou@lancaster.ac.uk)}

\date{}
\maketitle
\begin{abstract}
\noindent
In the literature there exists analytical expressions for the probability of a receiver decoding a transmitted source message that has been encoded using random linear network coding. In this work, we look into the probability that the receiver will decode at least a fraction of the source message, and present an exact solution to this problem for both non-systematic and systematic network coding. Based on the derived expressions, we investigate the potential of these two implementations of network coding for information-theoretic secure communication and progressive recovery of data. 
\\ \\
\noindent
\textbf{Keywords:} Random linear network coding, rank-deficient decoding, probability analysis, information-theoretic security.
\end{abstract}


\section{Introduction}
\label{sec.intro}

Random linear network coding (RNLC) is the process of constructing coded packets, which are random linear combinations of source packets over a finite field \cite{Ho06}. If $k$ source packets are considered, decoding at a receiving node starts after $k$ linearly independent coded packets have been collected. The probability of recovering all of the $k$ source packets when at least $k$ coded packets have been received has been derived in \cite{Trullols-Cruces11}. However, the requirement for a large number of received coded packets before decoding can introduce undesirable delays at the receiving nodes. In an effort to alleviate this problem, \textit{rank-deficient} decoding was proposed in \cite{Yan2013} for the recovery of a subset of source packets when fewer than $k$ coded packets have been obtained. Whereas the literature on network coding defines \textit{decoding success} as the recovery of 100\% of the source packets with a certain probability, the authors of \cite{Yan2013} presented simulation results that measured the \textit{fraction of decoding success}, that is, the recovery of a percentage of the source packets with a certain probability.

The fundamental problem that has motivated our work is the characterization of the probability of recovering some of the $k$ source packets when $n$ coded packets have been retrieved, where $n$ can be smaller than, equal to or greater than $k$. 
This idea was considered in~\cite{Gadouleau2011} for random network communications over a matroid framework. The authors show that partial decoding is highly unlikely. This problem has also been explored in the context of secure network coding, e.g., \cite{Bhattad2005,Lima2007}. Strict information-theoretic security can be achieved if and only if the mutual information between the packets available to an eavesdropper and the source packets is zero \cite{Cai2002}. When network coding is used, \textit{weak} security can be achieved if the eavesdropper cannot obtain $k$ linearly independent coded packets and, hence, cannot recover any meaningful information about the $k$ source packets \cite{Bhattad2005}. The authors of \cite{Bhattad2005} obtained bounds on the probability of RLNC being weakly secure and showed that the adoption of large finite fields improves security. A different setting but a similar problem was investigated in \cite{Lima2007}. Intermediate relay nodes between transmitting and receiving nodes were treated as potentially malicious, and criteria for characterizing the algebraic security of RLNC were defined. The authors demonstrated that the probability of an intermediate node recovering a strictly positive number of source packets tends to zero as the field size and the number of source packets go to infinity.

This paper revisits the aforementioned problem and obtains an exact expression for the probability that a receiving node will recover at least $x$ of the $k$ source packets if $n$ coded packets are collected, for $x\leq n$. The derived expression can be seen as a generalization of \cite[eq. (7)]{Trullols-Cruces11}. The \mbox{paper also looks} at the impact of transmitting source packets along with coded packets, known as \textit{systematic} RLNC, as opposed to transmitting only coded packets, referred to as \textit{non-systematic} RLNC.

In the remainder of the paper, Section~\ref{sec.System} formulates the problem, Section~\ref{sec.Analysis} obtains the probability of recovering a fraction of a network-coded message, Section~\ref{sec.results} presents results and Section~\ref{sec.conclusions} summarizes the conclusions of this work.


\section{System model and problem formulation}
\label{sec.System}

We consider a receiving network node, which collects $n$ packets and attempts to reconstruct a message that consists of $k$ source packets. The $n$ packets could have been broadcast by a single transmitting node or could have been originated from multiple nodes that possess the same message.

In the case of \textit{non-systematic} communication, transmitted packets are generated from the $k$ source packets using RLNC over $\mathbb{F}_q$ \cite{Ho06}, where $q$ is a prime power and $\mathbb{F}_q$ denotes the finite field of $q$ elements. In the case of \textit{systematic} RLNC, a sequence of $n_\mathrm{T}$ transmitted packets consists of the $k$ source packets and $n_\mathrm{T}-k$ coded packets that have been generated as in the non-systematic case. In both cases, a coding vector of length $k$, which contains the weighting coefficients used in the generation of a packet, is transmitted along with each packet. At the receiving node, the coding vectors of the $n$ successfully retrieved packets form the rows of a matrix $\mathbf{M}\in\mathbb{F}^{n\times k}_q$, where $\mathbb{F}_q^{n \times k}$ denotes the set of all $n \times k$ matrices over $\mathbb{F}_q$. The $k$ source packets can be recovered from the $n$ received packets if and only if $k$ of the $n$ coding vectors are linearly independent, implying that $\mathrm{rank}(\mathbf{M})=k$ for $n\geq k$. The probability that the $n\times k$ random matrix $\mathbf{M}$ has rank $k$ and, thus, the receiving node can reconstruct the entire message is given in \cite{Trullols-Cruces11} for non-systematic RLNC and \cite{Shrader09} for systematic RLNC.

The objective of this paper is to derive the probability that a receiving node will reconstruct at least $x\leq k$ source packets upon reception of $n$ network-coded packets. To formulate this problem, let $\mathbf{e}_i$ denote the $i$-th unit vector of length $k$.
A coding vector, or a row of $\mathbf{M}$, equal to $\mathbf{e}_i$ represents the $i$-th source packet. Let $X$ be the set of indices corresponding to the unit vectors contained in the rowspace of $\mathbf{M}$, denoted by $\Row(\mathbf{M})$, so that $X=\{i:\mathbf{e}_i \in \Row(\mathbf{M})\}$. We write $\lvert X\rvert$ to denote the cardinality of random variable $X$. Furthermore, we define random variables $R$ and $N$ to give the rank of $\mathbf{M}$ and the number of rows in $\mathbf{M}$, respectively. The considered problem has been decomposed into the following two tasks:
\begin{enumerate}
\item Obtain the probability of recovering at least $x$ source packets, provided that $r$ out of the $n$ received packets are linearly independent, for $x\leq r\leq k$. This is equivalent to finding the probability of $\Row(\mathbf{M})$ containing at least $x$ unit vectors, given $\mathbf{M}$ has $n$ rows and rank $r$. We denote this probability by $P(\lvert X\rvert\geq x\,\vert\,R=r,\,N=n)$.
\item {Obtain the probability of recovering at least $x$ source packets, provided that $n\geq x$ packets have been collected.} We write $P(\lvert X\rvert\geq x\,\vert\,N=n)$ to refer to this probability.
\end{enumerate}

Derivation of $P(\lvert X\rvert\geq x\,\vert\,R=r,\,N=n)$ and \mbox{$P(\lvert X\rvert\geq x\,\vert\,N=n)$} is the focus of the following section.


\section{Probability analysis}
\label{sec.Analysis}

The analysis presented in this section relies on the well- known Principle of Inclusion and Exclusion \cite[Prop. 5.2.2]{Cameron1994}, which is repeated below for clarity.
\begin{lemma} 
\label{lemma.PIE}
\textbf{Principle of inclusion and exclusion.}
Given a set $A$, let $f$ be a real valued function defined for all sets $S,J\subseteq A$. If $g(S) = \sum_{J:J\supseteq S} f(J)$ then $f(S) = \sum_{J:J\supseteq S} (-1)^{\lvert J\setminus S\rvert} g(J)$.
\end{lemma}
For non-negative integers $m$ and $d$, we denote by $\binom{m}{d}$ the binomial coefficient, which gives the number of $d$-element sets of an $m$-element set. 
The $q$-analog of the binomial coefficient, known as the \textit{Gaussian binomial coefficient} and denoted by $\qbinom{m}{d}$, enumerates all $d$-dimensional subspaces of an $m$-dimensional space over $\mathbb{F}_q$ \cite[p. 125]{Cameron1994}.

Given $\mathbf{M}$ has rank $r$, let $P(\lvert X\rvert=x\,\vert\,R=r,\,N=n)$ denote the probability of recovering \textit{exactly} $x \leq r$ source packets or, equivalently, the probability of $\Row(\mathbf{M})$ containing \textit{exactly} $x \leq r$ unit vectors. The following theorem obtains an expression for $P(\lvert X\rvert= x\,\vert\,R=r,\,N=n)$, which is then used in the derivation of $P(\lvert X\rvert\geq x\,\vert\,R=r,\,N=n)$.

\begin{theorem} 
\label{thm.cond_exactly_x}
Given a random $n \times k$ matrix $\mathbf{M}$ of rank $r$, the probability that the rowspace of $\mathbf{M}$ contains exactly $x \leq r$ unit vectors is given by
\begin{equation}
\label{eq.cond_exactly_x}
P(\lvert X\rvert=x\,\vert\,R=r,\,N=n)=\frac{\binom{k}{x}}{\qbinom{\phantom{\widetilde k}\!\!\!k}{r}}%
\displaystyle\sum_{j=0}^{k-x}(-1)^{j}\binom{k-x}{j}\!\genfrac{[}{]}{0pt}{}{\,k-x-j\,}{\,r-x-j\,}_{\!q}.
\end{equation}
\end{theorem}
\begin{proof}
For $S\!\subseteq\!J\subseteq\!\{1, \dots k \}$, let $g(S)$~be~the~probability that $\{ \mathbf{e}_i : i \in S \} \subseteq \Row(\mathbf{M})$, that is, the probability that $S \subseteq X$. This is just the probability that $\Row( \mathbf{M} )$ contains a fixed $\lvert S\rvert$-dimensional subspace, namely the space \mbox{$V=\Span\{\mathbf{e}_i:i \in S\}$}. We see that, by considering the quotient space $\mathbb{F}_q^k/V$, there is a direct correspondence between \mbox{$r$-dimensional} subspaces of $\mathbb{F}_q^k$ containing $V$, and $(r-\lvert S\rvert)$-dimensional subspaces of a \mbox{$(k-\lvert S\rvert)$-dimensional} space. Hence, there are $\qbinom{k-\lvert S \rvert}{r-\lvert S \rvert}$ $r$-dimensional subspaces of $\mathbb{F}_q^k$ containing $V$. The probability that $\Row(\mathbf{M})$ contains the space $V$ is equal to
\begin{equation} 
\label{g(S)}
g(S) = \frac{\qbinom{k-\lvert S \rvert}{r-\lvert S \rvert}}{\qbinom{\phantom{\widetilde k}\!\!\!k}{r}}.
\end{equation}
where the denominator in \eqref{g(S)} enumerates the $r$-dimensional subspaces of $\mathbb{F}_q^k$. Now, let $f(S)$ be the probability that $S=X$, that is, the probability that $\{ \mathbf{e}_i : i \in S \} \subseteq \Row( \mathbf{M})$ and $\mathbf{e}_i \notin \Row(\mathbf{M})$ for $i\notin S$. It follows that $g(S) = \sum_{J \supseteq S} f(J)$. Invoking the Principle of Inclusion and Exclusion (Lemma~\ref{lemma.PIE}) and using \eqref{g(S)}, we can write $f(S) = \sum_{J \supseteq S} (-1)^{|J\setminus S|} \cdot g(J)$ and expand it to
\begin{align}
f(S) \notag&= \sum_{J \supseteq S} (-1)^{|J\setminus S|}
\cdot
\frac{\qbinom{k-\lvert J \rvert}{r-\lvert J \rvert}}{\qbinom{k}{r}} \notag
\\
&=
\frac{1}{\qbinom{k}{r}}
\sum_{J' \subseteq \{ 1, \dots , k \}\setminus S}
(-1)^{\lvert J' \rvert}
\qbinom{k-\lvert S\rvert-\lvert J'\rvert}{r-\lvert S\rvert-\lvert J'\rvert}
 \label{f(S):J'}
\\
&=
\frac{1}{\qbinom{k}{r}}
\sum_{j=0}^{k-\lvert S\rvert}
(-1)^{j}
\binom{k-\lvert S\rvert}{j}
\qbinom{k-\lvert S\rvert-j}{r-\lvert S\rvert-j} \label{f(S):j}
\end{align}
where \eqref{f(S):J'} follows by setting $J'=J\setminus S$, and \eqref{f(S):j} follows since there are $\binom{k-\lvert S\rvert}{j}$ sets $J'$ of size $j$. Considering that $f(S)$ is the probability that $X=S$, we can write
\begin{equation}
P(\lvert X\rvert=x\,\vert\,R=r,\,N=n) =\!\!\sum_{S: \lvert S \rvert =x}\!\!\!f(S) = \binom{k}{x} f(S') \label{eq:P(|X|=x)}
\end{equation}
where $S'$ is any subset of $\{1, \dots, k \}$ of size $x$. The second equality in \eqref{eq:P(|X|=x)} holds since there are $\binom{k}{x}$ sets \mbox{$S \subseteq \{1, \dots, k\}$} of size $x$. Substituting \eqref{f(S):j} in \eqref{eq:P(|X|=x)} gives the result.
\end{proof}

\begin{remark}
Theorem~\ref{thm.cond_exactly_x} can be seen as a special case of~\cite[Proposition 6]{Gadouleau2011}. Whereas the proof in \cite{Gadouleau2011} uses elements of matroid theory, our paper proposes an alternative and more intuitive proof strategy. 
\end{remark}

\begin{corollary} 
\label{cor.cond_atleast_x}
Given a random $n \times k$ matrix $\mathbf{M}$ of rank $r$, the probability that the rowspace of $\mathbf{M}$ contains at least $x \leq r$ unit vectors is given by
\begin{equation}
\label{eq.cond_atleast_x}
P(\lvert X\rvert\!\geq\! x\,\vert\,R\!=\!r,\,N\!=\!n)
=\frac{1}{\qbinom{\phantom{\widetilde k}\!\!\!k}{r}}%
\!\sum_{i=x}^{r}\!\binom{k}{i}\!\sum_{j=0}^{k-i}(-1)^{j}\binom{k-i}{j}\!\qbinom{k-i-j}{r-i-j}.
\end{equation}
\end{corollary}
\begin{proof}
By definition, the probability $P(\lvert X\rvert\geq x\,\vert\,R=r,\,N=n)$ is equal to $\sum_{i=x}^{r} P(\lvert X\rvert=i\,\vert\,R=r,\,N=n)$. Substituting in \eqref{eq.cond_exactly_x} gives the result.
\end{proof}

Note that, although $\mathbf{M}$ is an $n \times k$ matrix, the probabilities in \eqref{eq.cond_exactly_x} and \eqref{eq.cond_atleast_x} hold for any value of $n\!\geq\!r$. Having obtained an expression for the probability $P(\lvert X\rvert\!\geq\!x\,\vert\,R\!=\!r,\,N\!=\!n)$, we now proceed to the derivation of $P(\lvert X\rvert\!\geq\!x\,\vert\,N\!=\!n)$. This probability is denoted by $P_\mathrm{ns}(\lvert X\rvert\!\geq\!x\,\vert\,N\!=\!n)$ and $P_\mathrm{s}(\lvert X\rvert\!\geq\! x\,\vert\,N\!=\!n)$ for non-systematic and systematic RLNC, respectively. Expressions for each case are derived in the following two propositions.

\begin{proposition}
\label{prop.nonsys}
If a receiving node collects $n$ random linear combinations of $k$ source packets, the probability that at least $x\leq k$ source packets will be recovered is
\begin{align}
\label{eq.prob_ns}
&P_\mathrm{ns}(\lvert X\rvert \geq x \,\vert\,N=n)\nonumber\\%
&=\frac{1}{q^{nk}}
\displaystyle\sum_{r=x}^{\min(n,k)}\!%
\Biggl(%
\displaystyle\sum_{i=x}^{r}%
\binom{k}{i}%
\displaystyle\sum_{j=0}^{k-i}%
\left(-1\right)^{j}%
\binom{k-i}{j}\!%
\genfrac{[}{]}{0pt}{}{\,k-i-j\,}{\,r-i-j\,}_{\!q\!}%
\Biggr)\!%
\displaystyle\prod_{\ell=0}^{r-1}(q^n-q^\ell).%
\end{align}
\end{proposition}
\begin{proof}
Let $P(R\!=\!r\,\vert\,N\!=\!n)$ denote the probability that the $n\times k$ matrix $\mathbf{M}$ has rank $r$. This is equivalent to the probability that $r$ out of the $n$ collected packets are linearly independent. The probability that at least $x$ of the $k$ source packets will be recovered can be obtained from
\begin{align}
\label{eq.ns_step1}
&P_\mathrm{ns}(\lvert X\rvert\geq x \,\vert\,N=n)\nonumber\\
&=\!\!\!\sum_{r=x}^{\min(n,k)}\!\!P(R=r\,\vert\,N=n)\,P(\lvert X\rvert\geq x\,\vert\,R=r,N=n).
\end{align}
The probability $P(R=r\,\vert\,N\!=\!n)$ is equal to \cite[Sec. II.A]{Gadouleau2010}
\begin{equation}
\label{eq.ns_step3}
P(R=r\,\vert\,N\!=\!n) = \frac{1}{q^{nk}}\qbinom{n}{r}%
\prod_{\ell=0}^{r-1}\left(q^{k}-q^{\ell}\right).%
\end{equation}
Substituting \eqref{eq.cond_atleast_x} and \eqref{eq.ns_step3} into \eqref{eq.ns_step1} and taking into account that
\begin{equation}
\label{eq.ns_step4}
\frac{\qbinom{n}{r}}{\qbinom{k}{r}}\prod_{\ell=0}^{r-1}(q^k-q^\ell)=\prod_{\ell=0}^{r-1}(q^n-q^\ell)
\end{equation}
leads to \eqref{eq.prob_ns}.
\end{proof}

\newpage
\begin{proposition}
\label{prop.sys}
If $k$ source packets and $n_\mathrm{T}-k$ random linear combinations of those $k$ source packets are transmitted over single-hop links, the probability that a receiving node will recover at least $x\leq k$ source packets from $n\leq n_\mathrm{T}$ received packets is
\begin{align}
\label{eq.prob_s}
&P_\mathrm{s}(\lvert X\rvert\geq x\,\vert\, N=n )%
\nonumber%
\\
&=\frac{1}{\binom{n_\mathrm{T}}{n}}
\cdot\!\!\displaystyle\sum_{r=x}^{\min(n,k)}\!\!\!%
\sum_{h=h_{\min}}^{r}\!\!\Biggl(%
\binom{k}{h}\binom{n_\mathrm{T}-k}{n-h}%
q^{-(n-h)(k-h)}%
\!\!\prod_{\ell=0}^{r-h-1}%
(q^{n-h}-q^\ell)%
\cdot\nonumber%
\\
&\quad\cdot\!\!\displaystyle\sum_{i=x_{\min}}^{r-h}%
\!\!\!\binom{k-h}{i}%
\!\sum_{j=0}^{k-h-i}%
\left(-1\right)^{j}%
\binom{k-h-i}{j}%
\qbinom{k-h-i-j}{r-h-i-j}%
\Biggr)%
\end{align}
where $h_{\min}\!=\!\max{(0,n\!\,-\!\,n_\mathrm{T}\!\,+\!\,k)}$ and $x_{\min}\!=\!\max(0,x\!-\!h)$.
\end{proposition}
\begin{proof}
Let us assume that some or none of the $k$ transmitted source packets have been received and let $X^{\prime}\subseteq X$ be the set of indices of the remaining source packets that can be recovered from the received coded packets. If $n^{\prime}$ of the $n_\mathrm{T}-k$ coded packets have been received and $k^{\prime}$ source packets remain to be recovered, the respective coding vectors will form an $n^{\prime}\times k^{\prime}$ random matrix $\mathbf{M}^{\prime}$. The probability that $r^{\prime}\leq\min(k^{\prime}, n^{\prime})$ coding vectors are linearly independent and at least $x^{\prime}\leq r^{\prime}$ source packets can be recovered is given by
\begin{multline}
\label{eq.s_step1}
P(\lvert X^{\prime}\rvert\geq x^{\prime}, R^{\prime}=r^{\prime}\,\vert\,N^{\prime}=n^{\prime})=
\\
P(R^{\prime}=r^{\prime}\,\vert\,N^{\prime}=n^{\prime})\,%
P(\lvert X^{\prime}\rvert\geq x^{\prime}\,\vert\, R^{\prime}=r^{\prime},\,N^{\prime}=n^{\prime})
\nonumber
\end{multline}
where the two terms of the product can be obtained from \eqref{eq.ns_step3} and \eqref{eq.cond_atleast_x}, respectively. The random variables $N'$ and $R'$ denote the number of received coded packets and the rank of matrix $\mathbf{M}^{\prime}$, respectively. If $n$ of the $n_\mathrm{T}$ transmitted packets are received, the probability that $h$ of them are source packets and the remaining $n-h$ are coded packets is
\begin{equation}
\label{eq.s_step2}
P(N^{\prime}=n-h\,\vert\, N=n)=\frac{\binom{k}{h}\binom{n_\mathrm{T}-k}{n-h}}{\binom{n_\mathrm{T}}{n}}.%
\end{equation}
The coding vectors of the $n$ received packets compose a matrix of rank $r$, based on which $x$ or more source packets can be recovered when $h$ of the $n$ received packets are source packets. Parameters $x^{\prime}$, $r^\prime$, $k^{\prime}$ and $n^{\prime}$, which are concerned with the received \textit{coded} packets only, can be written as $x-h$, $r-h$, $k-h$ and $n-h$, respectively. The probability of recovering at least $x$ source packets for all valid values of $r$ and $h$ is
\begin{align}
\label{eq.s_step3}
P&_\mathrm{s}(\lvert X\rvert\geq x \,\vert\, N=n)\nonumber
\\
=&%
\sum_{r=x}^{\min(n,k)}\!\!%
\sum_{h=h_{\min}}^{r}\!\!%
P(N^{\prime}=n-h\,\vert\, N=n)\,\cdot\nonumber%
\\
&\cdot P\bigl(\lvert X^{\prime}\rvert \geq\max(0,\,x\!-\!h),R^{\prime}=r-h\,\vert\,N^{\prime}=n\!-\!h\bigr)
\end{align}
which expands into \eqref{eq.prob_s}. Note that $\max(0,\,x-h)$ ensures that the value of $\lvert X^{\prime}\rvert$ is a non-negative integer when $h>x$.
\end{proof}

\begin{remark}
In systematic RLNC, if the receiving node attempts to recover source packets as soon as the transmission is initiated, i.e., $n_\mathrm{T}\!\leq\! k$, at least $x$ source packets will certainly be recovered when $n\!\geq\! x$ source packets are received, that is,
\begin{equation}
\label{eq.prob_s_less_than_k}
P_\mathrm{s}(\lvert X\rvert\geq x\,\vert\, N=n)=\left\{%
\begin{array}{ll}
1,&\mathrm{if}\;\; n_\mathrm{T}\leq k\;\;\mathrm{and}\;\;x \leq n\\
0,&\mathrm{if}\;\; n_\mathrm{T}\leq k\;\;\mathrm{and}\;\;x > n.%
\end{array}\right.
\end{equation}
\end{remark}


\section{Results and discussion}
\label{sec.results}

\begin{figure}[t]
\centering
\includegraphics[width=0.95\columnwidth]{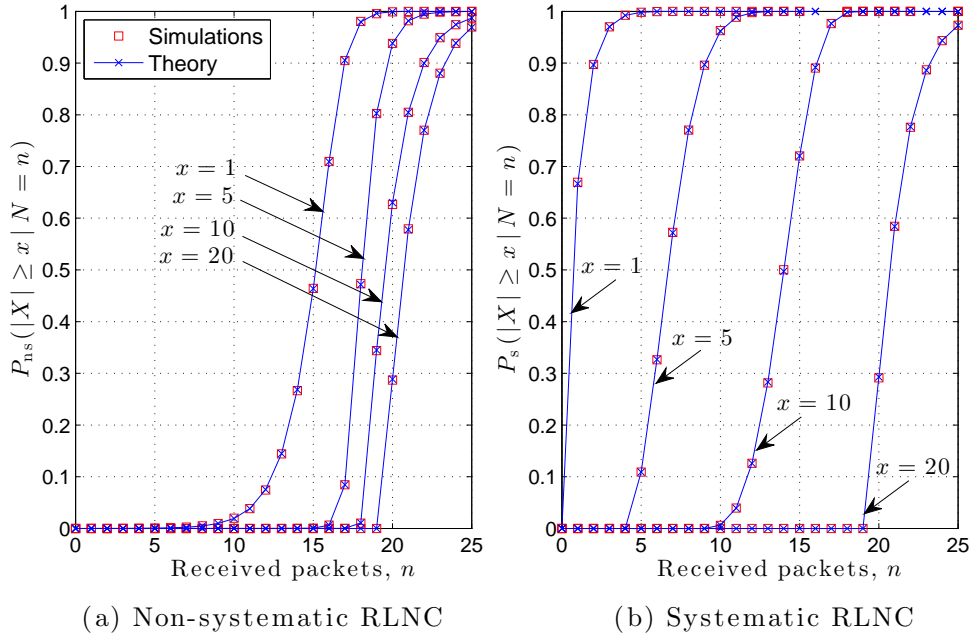}
\caption{{Simulation results and theoretical values for (a) non-systematic RLNC and (b) systematic RLNC. The probability of recovering at least $x$ source packets has been plotted for $q=2$, $k=20$, $x=1,5,10,20$ and $n_\mathrm{T}=30$.}}
\label{fig.theory_vs_sims}
\end{figure}

In order to demonstrate the exactness of the derived expressions, simulations that generated 60000 realisations of an $n\times k$ random matrix $\mathbf{M}$ over $\mathbb{F}_2$ were carried out for $n=1,\dots,30$ and $k\!=\!20$. In each case, matrix $\mathbf{M}$ was converted into reduced row echelon form using Gaussian elimination. Then, the rows that correspond to unit vectors $\mathbf{e}_i$, which represent recoverable source packets, were counted and averaged over all realisations. Fig.~\ref{fig.theory_vs_sims}(a) and Fig.~\ref{fig.theory_vs_sims}(b) show that measurements obtained through simulations match the calculations obtained from \eqref{eq.prob_ns} and \eqref{eq.prob_s} for non-systematic RLNC and systematic RLNC, respectively. In general, simulation results match  analytical predictions for any finite field $\mathbb{F}_q$ of order $q\geq 2$.

Fig.~\ref{fig.study_both} considers the simple case of RLNC transmission over a broadcast erasure channel. If the transmission of $n_\mathrm{T}$ packets is modeled as a sequence of $n_\mathrm{T}$ Bernoulli trials whereby $\varepsilon$ signifies the probability that a transmitted packet will be erased, the probability that a receiving node shall recover \textit{at least} $x$ of the $k$ source packets can be expressed as
\begin{equation}
\label{eq.prob_BEC}
P(\lvert X\rvert\!\geq\! x)\!=\!\displaystyle\sum_{n=x}^{n_\mathrm{T}}\!\displaystyle\binom{n_\mathrm{T}}{n}\!\left(1-\varepsilon\right)^{n}\varepsilon^{n_\mathrm{T}-n}\,P(\lvert X\rvert\!\geq\! x\,\vert\, N\!=\!n).
\end{equation}
The probability $P(\lvert X\rvert\geq x\,\vert\, N=n)$ is equal to \eqref{eq.prob_ns} for non-systematic RLNC {and \eqref{eq.prob_s} or \eqref{eq.prob_s_less_than_k}}, depending on the value of $n_\mathrm{T}$, for systematic RLNC. 

Fig.~\ref{fig.study_both}(a) focuses on non-systematic RLNC and depicts $P(\lvert X\rvert\geq x)$ in terms of $n_\mathrm{T}$ for $x\in\{2,4,10,16,20\}$ when $k\!=\!20$, and for $x\in\{3,6,15,24,30\}$ when $k\!=\!30$. Results have been obtained for $q\in\{2,8\}$ and $\varepsilon\!=\!0.2$. For $q\!=\!2$, the transmission of only a few additional coded packets can increase the fraction of the recovered message from at least $x/k=0.1$ to $x/k=1$. However, for $q$ as low as $8$, the range of $n_\mathrm{T}$ values for which a receiving node will proceed from recovering a small portion of the transmitted message to recovering the whole message gets very narrow. Furthermore, for $q\!=\!2$, segmentation of the message into $k\!=\!20$ source packets permits a receiving node to recover the same fraction ($x/k$) of the message with a higher probability than dividing the same message into $k\!=\!30$ source packets.

Systematic RLNC is considered in Fig.~\ref{fig.study_both}(b). Besides the reduced decoding complexity \cite{Lucani12}, we observe that systematic RLNC enables a receiving node to gradually reveal an increasingly larger portion of the message as more packets are transmitted. However, a large number of source packets or a high order finite field impairs the progressive recovery of the message for $n_\mathrm{T}\!>\!k$. This is because source packets are transmitted for $n_\mathrm{T}\leq k$ but coded packets are sent for $n_\mathrm{T}\!>\!k$; the decoding behaviour of a receiving node changes at $n_\mathrm{T}\!=\!k$ and causes a change in the slope of $P(\lvert X\rvert\geq x)$ for $x/k\!=\!0.8$.

The results show that, if information-theoretic security is required, non-systematic RLNC over finite fields of size 8 or larger can be used to segment each message into a large number of source packets. The number of transmitted packets can then be adjusted to the channel conditions to achieve a balance between the probability of legitimate nodes reconstructing the message and the probability of eavesdroppers being unable to decode even a portion of the message. If the objective of the system is to maximize the number of nodes that will recover at least a large part of a message, systematic RLNC over small finite fields can be used to divide data into source packets. If the receiving nodes do not suffer from limited computational capabilities, the size of the finite field can be increased to improve the probability of recovering the entire message.

\begin{figure}[t]
\captionsetup[subfigure]{labelformat=empty}
\begin{subfigure}{0.492\columnwidth}
\centering
\includegraphics[height=2.4cm]{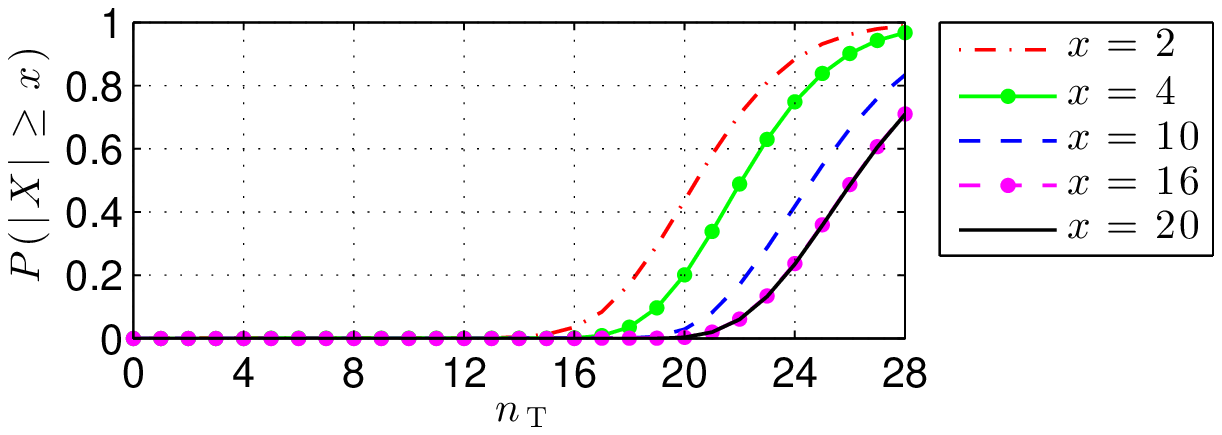}
\vspace{-3mm}
\caption{\scriptsize $k=20$, $q=2$}
\end{subfigure}
\begin{subfigure}{0.492\columnwidth}
\centering
\includegraphics[height=2.4cm]{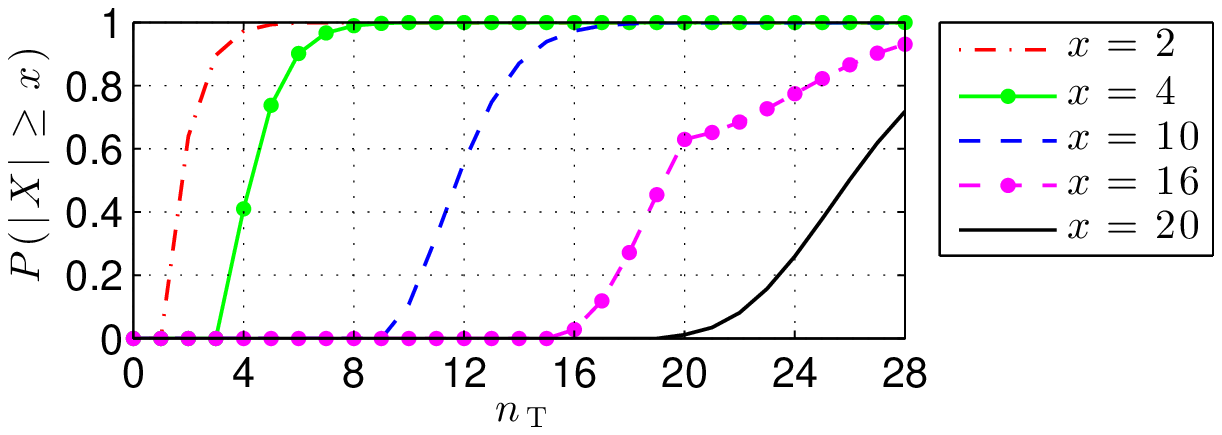}
\vspace{-3mm}
\caption{\scriptsize $k=20$, $q=2$}
\end{subfigure}
\\
\begin{subfigure}{0.492\columnwidth}
\vspace{1mm}
\centering
\includegraphics[height=2.4cm]{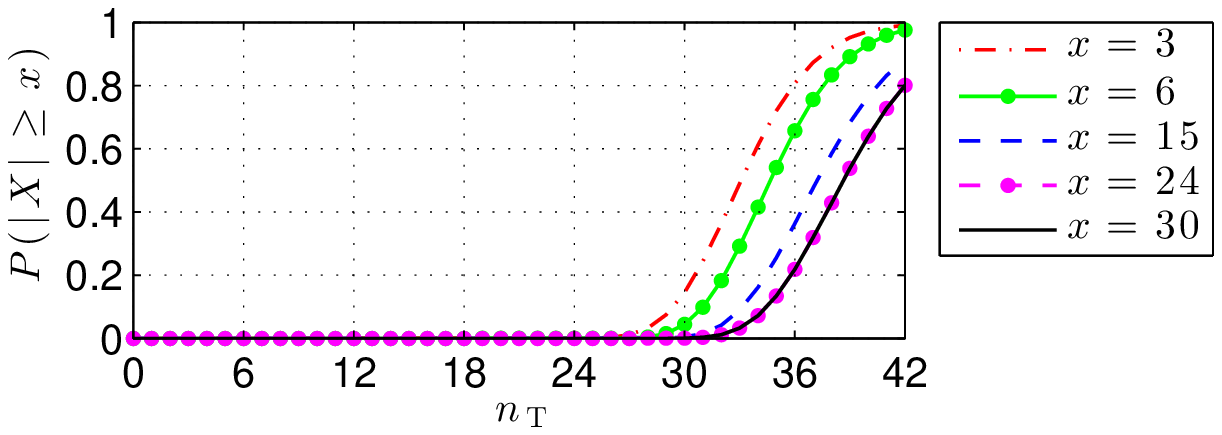}
\vspace{-3mm}
\caption{\scriptsize $k=30$, $q=2$}
\end{subfigure}
\begin{subfigure}{0.492\columnwidth}
\vspace{1mm}
\centering
\includegraphics[height=2.4cm]{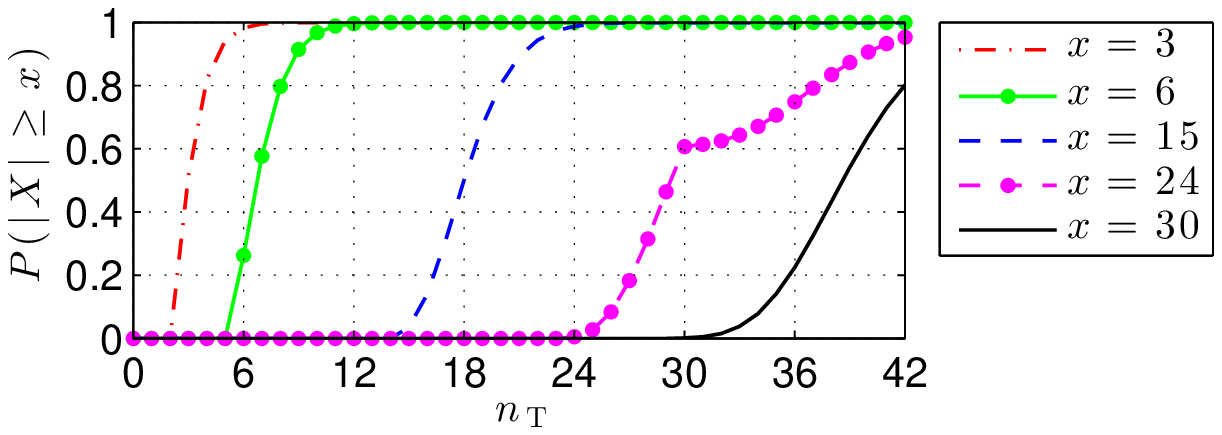}
\vspace{-3mm}
\caption{\scriptsize $k=30$, $q=2$}
\end{subfigure}
\\
\begin{subfigure}{0.492\columnwidth}
\vspace{1mm}
\centering
\includegraphics[height=2.4cm]{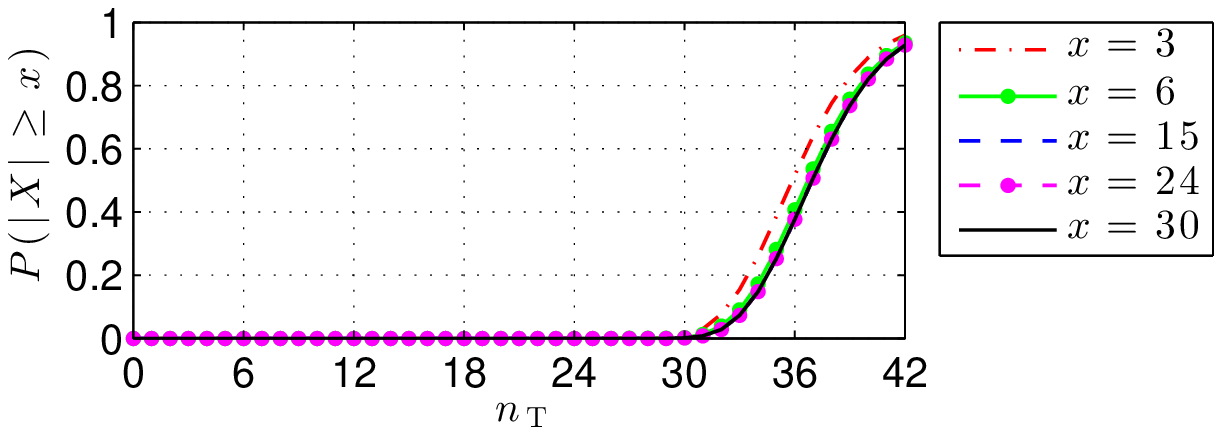}
\vspace{-3mm}
\caption{\scriptsize $k=30$, $q=8$}
\end{subfigure}
\begin{subfigure}{0.492\columnwidth}
\vspace{1mm}
\centering
\includegraphics[height=2.4cm]{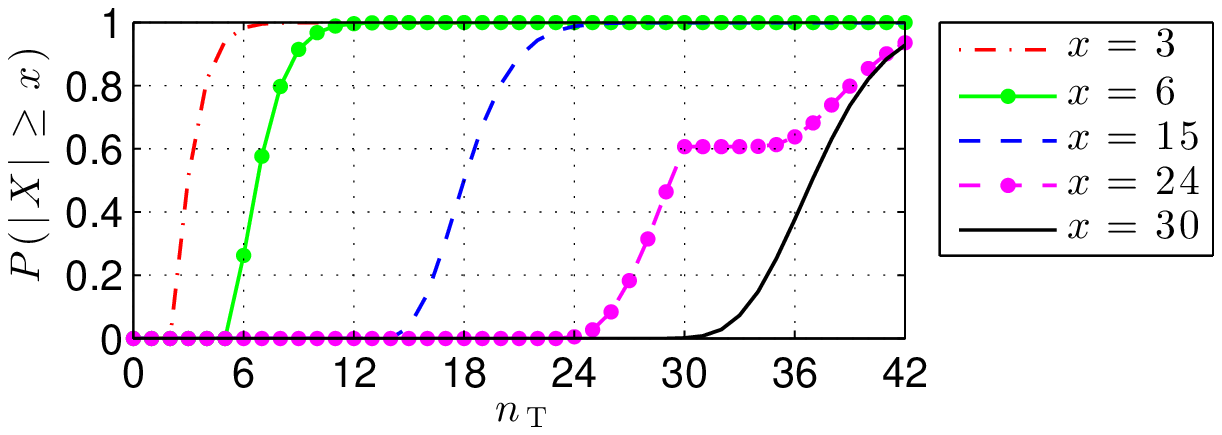}
\vspace{-3mm}
\caption{\scriptsize $k=30$, $q=8$}
\end{subfigure}
\\
\begin{subfigure}{0.492\columnwidth}
\vspace{2mm}
\centering
\caption{\scriptsize (a) Non-systematic RLNC}
\end{subfigure}
\begin{subfigure}{0.492\columnwidth}
\vspace{2mm}
\centering
\caption{\scriptsize (b) Systematic RLNC}
\end{subfigure}
\vspace{-3mm}
\caption{Depiction of the probability of recovering at least $x$ source packets when $n_\mathrm{T}$ packets have been transmitted over a packet erasure channel with $\varepsilon\!=\!0.2$ using (a) non-systematic RLNC and (b) systematic RLNC.}
\label{fig.study_both}
\vspace{-4mm}
\end{figure}


\section{Conclusions}
\label{sec.conclusions}

This paper derived exact expressions for the probability of decoding a fraction of a source message upon reception of an arbitrary number of network-coded packets. Results unveiled the potential of non-systematic network coding in offering weak information-theoretic security, even when operations are over small finite fields. On the other hand, systematic network coding allows for the progressive recovery of the source message as the number of received packets increases, especially when the size of the finite field is small.


\section{Acknowledgments}
\label{sec.ack}

Jessica Claridge has been supported by an EPSRC PhD studentship. Both authors appreciate the support of the COST Action IC1104 and thank Simon R. Blackburn for his advice.


\bibliographystyle{IEEEtran}
\bibliography{IEEEabrv,references}

\begin{thebibliography}{10}
\providecommand{\url}[1]{#1}
\csname url@samestyle\endcsname
\providecommand{\newblock}{\relax}
\providecommand{\bibinfo}[2]{#2}
\providecommand{\BIBentrySTDinterwordspacing}{\spaceskip=0pt\relax}
\providecommand{\BIBentryALTinterwordstretchfactor}{4}
\providecommand{\BIBentryALTinterwordspacing}{\spaceskip=\fontdimen2\font plus
\BIBentryALTinterwordstretchfactor\fontdimen3\font minus
  \fontdimen4\font\relax}
\providecommand{\BIBforeignlanguage}[2]{{%
\expandafter\ifx\csname l@#1\endcsname\relax
\typeout{** WARNING: IEEEtran.bst: No hyphenation pattern has been}%
\typeout{** loaded for the language `#1'. Using the pattern for}%
\typeout{** the default language instead.}%
\else
\language=\csname l@#1\endcsname
\fi
#2}}
\providecommand{\BIBdecl}{\relax}
\BIBdecl

\bibitem{Ho06}
T.~Ho, M.~M\'{e}dard, R.~Koetter, D.~R. Karger, M.~Effros, J.~Shi, and
  B.~Leong, ``A random linear network coding approach to multicast,''
  \emph{{IEEE} Trans. Inf. Theory}, vol.~52, no.~10, pp. 4413--4430, Oct. 2006.

\bibitem{Trullols-Cruces11}
O.~Trullols-Cruces, J.~Barcelo-Ordinas, and M.~Fiore, ``Exact decoding
  probability under random linear network coding,'' \emph{{IEEE} Commun.
  Lett.}, vol.~15, no.~1, pp. 67--69, Jan. 2011.

\bibitem{Yan2013}
Z.~Yan, H.~Xie, and B.~W. Suter, ``Rank deficient decoding of linear network
  coding,'' in \emph{Proc. IEEE Int. Conf. Acoustics, Speech and Signal
  Processing}, Vancouver, BC, May 2013, pp. 5080--5084.

\bibitem{Gadouleau2011}
M.~Gadouleau and A.~Goupil, ``A matroid framework for noncoherent random
  network communications,'' \emph{{IEEE} Trans. Inf. Theory}, vol.~57, no.~2,
  pp. 1031--1045, Feb. 2011.

\bibitem{Bhattad2005}
K.~Bhattad and K.~R. Narayanan, ``Weakly secure network coding,'' in
  \emph{Proc. 1st Workshop on Network Coding, Theory and Applications}, Riva
  Del Garda, Italy, Apr. 2005.

\bibitem{Lima2007}
L.~Lima, M.~M\'{e}dard, and J.~Barros, ``Random linear network coding: {A} free
  cipher?'' in \emph{Proc. IEEE Int. Symp. Inform. Theory}, Nice, France, Jun.
  2007, pp. 546--550.

\bibitem{Cai2002}
N.~Cai and R.~W. Yeung, ``Secure network coding,'' in \emph{Proc. IEEE Int.
  Symp. on Inform. Theory}, Lausanne, Switzerland, Jun. 2002, p. 323.

\bibitem{Shrader09}
B.~Shrader and N.~M. Jones, ``Systematic wireless network coding,'' in
  \emph{Proc. IEEE Military Commun. Conf.}, Boston, MA, Oct. 2009.

\bibitem{Cameron1994}
P.~J. Cameron, \emph{Combinatorics: {T}opics, techniques, algorithms}.\hskip
  1em plus 0.5em minus 0.4em\relax Cambridge University Press, 1994.

\bibitem{Gadouleau2010}
M.~Gadouleau and Z.~Yan, ``Constant-rank codes and their connection to
  constant-dimension codes,'' \emph{{IEEE} Trans. Inf. Theory}, vol.~56, no.~7,
  pp. 3207--3216, Jul. 2010.

\bibitem{Lucani12}
D.~E. Lucani, M.~M\'{e}dard, and M.~Stojanovic, ``On coding for delay --
  {N}etwork coding for time-division duplexing,'' \emph{{IEEE} Trans. Inf.
  Theory}, vol.~58, no.~4, pp. 2330--2348, Apr. 2012.

\end{thebibliography}


\end{document}